\documentclass[conference,a4paper]{IEEEtran}
\usepackage{graphicx, amsfonts, amsmath, amssymb}
\IEEEoverridecommandlockouts
\usepackage{amsmath}%
\usepackage{amsfonts}%
\usepackage{amssymb}%
\usepackage{mathrsfs}%
\usepackage{texdraw}%
\usepackage{psfrag}%
\usepackage{accents}%
\usepackage{color}%
\usepackage{cite}
\usepackage{graphicx}%
\usepackage{graphics}
\usepackage[tight,footnotesize]{subfigure}

\newtheorem{theorem}{Theorem}
\newtheorem{lemma}{Lemma}

\newtheorem{remark}{Remark}
\newtheorem{cor}{Corollary}

\begin{document}

\title{ Gaussian Sensor Networks with Adversarial Nodes
\thanks{{This work was supported in part by the NSF under grants CCF-1016861,  CCF-1118075 and  CCF-1111342.}
}}

\author{
 \IEEEauthorblockN{  Emrah Akyol, Kenneth Rose &  Tamer Ba\c{s}ar}
\IEEEauthorblockA{\{eakyol, rose\}@ece.ucsb.edu  &  basar1@illinois.edu\\
UC, Santa Barbara  &  University of Illinois at Urbana-Champaign\\
}}

\maketitle

\begin{abstract}
This paper studies a particular sensor network model which involves one single Gaussian source observed by many sensors, subject to additive independent Gaussian observation noise. Sensors communicate with the receiver over an additive Gaussian multiple access channel. The aim of the receiver is to reconstruct the underlying source with minimum mean squared error. The scenario of interest here is one where some of the sensors act as adversary (jammer): they strive to maximize distortion. We show that the ability of transmitter sensors to secretly agree on a random event, that is ``coordination", plays a key role in the analysis. Depending on the coordination capability of sensors and the receiver,  we consider two problem settings.  The first setting involves transmitters with ``coordination" capabilities in the sense that all transmitters can use identical realization of randomized encoding for each transmission.  In this case, the optimal strategy for the adversary sensors also requires coordination, where they all generate the same realization of independent and identically distributed Gaussian noise.  In the second setting, the transmitter sensors are restricted to use fixed, deterministic encoders and this setting, which corresponds to a Stackelberg game, does not admit a saddle-point solution. We  show that the the optimal strategy for all sensors is uncoded communications where encoding functions of adversaries and transmitters are in opposite directions. For both settings, digital compression and communication is strictly suboptimal. 
\end{abstract}

\begin{keywords}Sensor networks, game theory, uncoded communication, analog mappings, coordinated transmission \end{keywords}

\IEEEpeerreviewmaketitle
\section{Introduction}

 Distributed compression and communication over sensor networks has been an important problem, see e.g. \cite{xiong2004distributed} for an overview. Joint source-channel coding (JSCC) has certain advantages over separate source and channel coding, and several specific aspects of JSCC over sensor networks have been studied in previous work; see e.g.,  \cite{gastpar2005power} and the references therein. In this paper, we extend the game theoretic analysis of the Gaussian test channel \cite{basar1983gaussian,basar1985complete,basar1986solutions,bansal1989communication} to Gaussian sensor networks introduced by \cite{gastpar2008uncoded}. A particular extension of \cite{gastpar2008uncoded} to asymmetric sensor networks was studied in \cite{behroozi2008optimal, behroozi2010optimal}. Communication games within general multiple input-multiple output settings were considered in \cite{kashyap2004correlated ,shafiee2009mutual}. 

In this paper, we consider two settings for the  sensor network model  illustrated in Figure 1 and explained in detail in Section II. The first $M$ sensors (i.e., the transmitters) and the receiver constitute Player 1 (minimizer) and the remaining $K$ sensors (i.e., the adversaries) constitute Player 2 (maximizer). This zero-sum game does not admit a saddle-point in pure strategies (fixed encoding functions), but admits one in mixed  strategies (randomized functions).  In the first setting, the transmitter sensors can use randomized encoders, i.e., all transmitters and the receiver agree on some (pseudo)random sequence, denoted as $\{\gamma\}$ in the paper. We coin the term ``coordination" for this capability and show that it has a pivotal role in the analysis and the implementation of optimal strategies for both the transmitter and adversarial sensors.   In the first setting, we consider the more general case of mixed strategies and present the saddle-point solution in Theorem 1. In the second setting, encoding functions of transmitters are limited to the fixed mappings. This setting can be viewed as a Stackelberg game where Player 1 is the leader, restricted to pure strategies, and Player 2 is the follower, who observes Player 1's choice of pure strategies and plays accordingly. We present  in Theorem 2 the optimal strategies for this Stackelberg game, whose cost is strictly higher than the cost associated with the first setting. The sharp contrast between the two settings underlines the importance of ``coordination" in sensor networks with adversarial nodes.



\section{Problem Definition}
 In general, lowercase letters (e.g., $x$) denote scalars, boldface lowercase (e.g., $\boldsymbol x$) vectors, uppercase (e.g., $U, X$) matrices and random variables, and boldface uppercase (e.g., $\boldsymbol X$) random
vectors.  $\mathbb E(\cdot)$,  $\mathbb P(\cdot)$ and  $\mathbb R$  denote the expectation and probability operators, and  the set of real numbers respectively. $Bern(p)$ denotes the Bernoulli random variable, taking values $1$ with probability $p$ and $-1$ with $1\!-\!p$. Gaussian distribution with mean $\boldsymbol \mu$ and covariance matrix $R$ is denoted as $\mathcal N(\boldsymbol \mu,R)$. 

The sensor network model is illustrated in Figure 1. The underlying source  $\{S(i)\}$ is a sequence of i.i.d. real valued Gaussian random variables with zero mean and variance $\sigma_S^2$. Sensor $m \in [1\!:\!M\!+\!K]$ observes a sequence $\{U_m(i)\}$  defined as 
\begin{equation}
U_m (i)= S(i)+W_m(i),  
\end{equation}  
where $\{W_m(i)\}$ is a sequence of i.i.d.  Gaussian random variables with zero mean and variance $\sigma_{W_m}^2$, independent of  $\{S(i)\}$.  Sensor $m \in [1\!:\!M\!+\!K]$ can apply arbitrary Borel measurable function $g_m^N:\mathbb R^N \rightarrow \mathbb R$ to the observation sequence of length $N$, $\boldsymbol U_m$ so as to generate  sequence of channel inputs $X_m(i)=g_m^N(\boldsymbol U_m)$ under power constraint: 
\begin{equation}
\lim \limits_{N\rightarrow \infty}\frac{1}{N} \sum \limits_{i=1}^{N} \mathbb E \{X_m^2(i)\} \leq P_m
\label{power}
\end{equation}
The channel output is then given as 
\begin{equation}
Y(i)=Z(i)+\sum \limits_{j=1}^{M+K} X_j(i)
\end{equation}
where $\{Z(i)\}$ is a sequence of i.i.d. Gaussian random variables of zero mean and variance $\sigma_Z^2$, independent of  $\{S(i)\}$ and  $\{W_m(i)\}$. The receiver applies a Borel measurable function $h^N: \mathbb R^N \rightarrow \mathbb R$ to the received sequence $\{Y(i)\}$ to minimize the cost, which is measured as mean squared error (MSE) between the underlying source $S$ and the estimate at the receiver $\hat S$ as
\begin{equation}
J(g_m^N(\cdot),h^{N}(\cdot))=\lim \limits_{N\rightarrow \infty} \frac{1}{N}\sum \limits_{i=1}^{N} \mathbb E \{(S(i)-{\hat S}(i))^2\}
\label{costf}
\end{equation}
for $m=1,2, \ldots, M+K$.


The transmitters $g_m^N(\cdot)$ for $m\in [1\!:\!M]$ and the receiver $h^N(\cdot)$ seek to minimize the cost (\ref{costf}) while the adversaries aim to maximize  (\ref{costf})  by properly choosing  $g_k^N(\cdot)$ for $k\in [M\!+\!1\!:\!M\!+\!K]$. We focus on the symmetric sensor and symmetric source  where $P_m=P_T$  and ${\sigma_W^2}_m={\sigma_W^2}_T$, $\forall m \in [1\!:\!M]$ and ${\sigma_{W_k}^2}={\sigma_{W_T}^2}$ and $P_k=P_A$,  $\forall k \in [M+1\!:\!M\!+\!K]$. 

A transmitter-receiver-adversarial policy ($g_m^{N*},g_k^{N*},h^{N*}$) constitutes a saddle-point solution  if it satisfies the pair of inequalities
\begin{equation}
\label{saddle}
J(g_m^{N*},g_k^{N},h^{N}) \leq J(g_m^{N*},g_k^{N*},h^{N*}) \leq J(g_m^{N},g_k^{N*},h^{N}) 
\end{equation}


\begin{figure}
\centering
\includegraphics[scale=0.45]{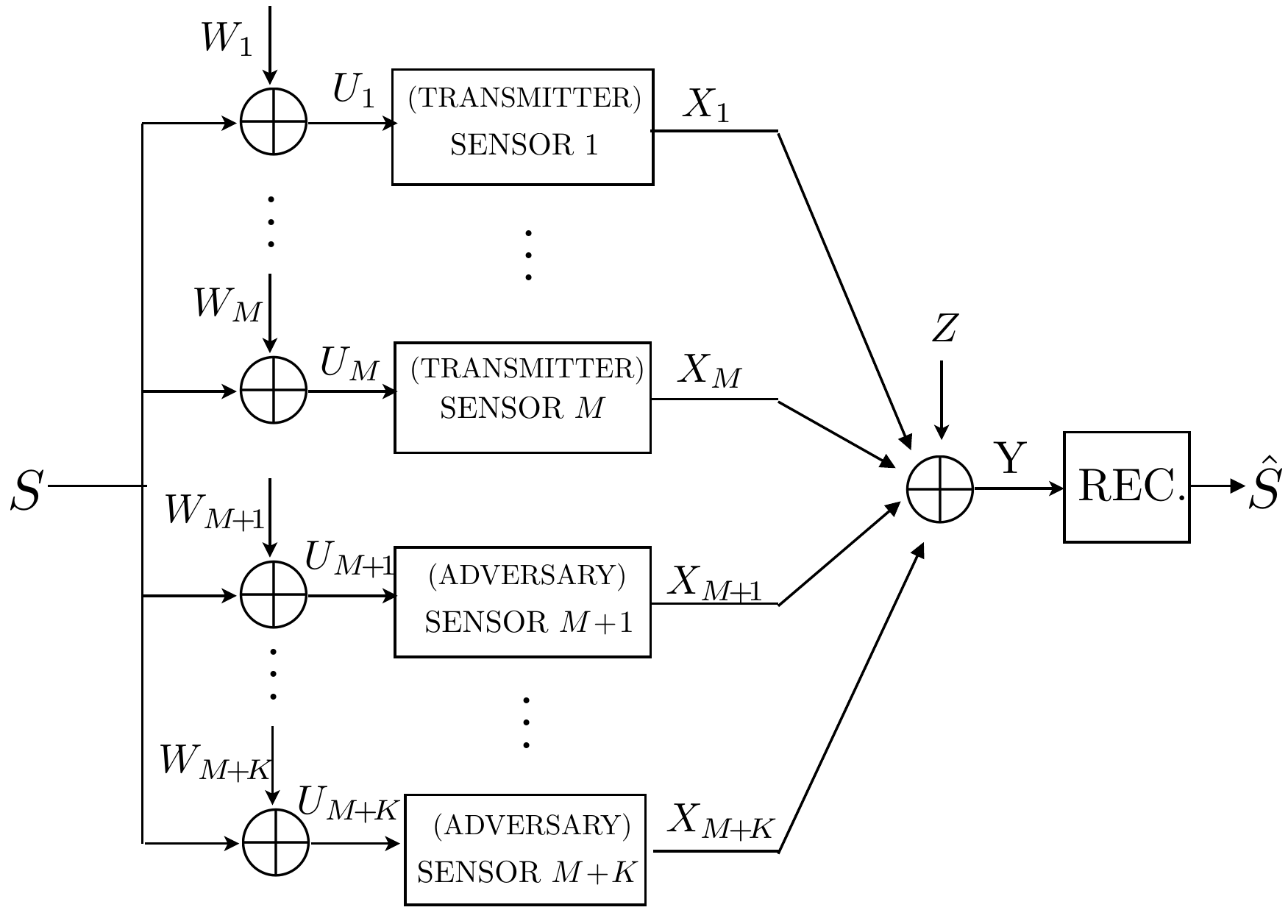}
\label{figure4}
\caption{The sensor network model.}
\end{figure}

\section{Problem Setting I}

The first scenario is concerned with the setting where the transmitter sensors have the ability to {\it coordinate}, i.e., all transmitters and the receiver can  agree on an i.i.d.  sequence of random variables $\{\gamma(i)\}$ generated, for example, by a side channel, the output of which is, however, not available to the adversarial sensors\footnote{An alternative practical method to coordinate is to generate the identical pseudo-random numbers  at each sensor, based on pre-determined seed.}. The ability of coordination allows transmitters and the receiver to agree on randomized encoding mappings. Surprisingly, in this setting, the adversarial sensors also need to coordinate, i.e., agree on an i.i.d. random sequence, denoted as $\{\theta(i)\}$, to generate the optimal jamming strategy.

The saddle-point solution of this problem is presented in the following theorem.

\begin{theorem}
Under scenario 1, the optimal encoding function for the transmitters is randomized uncoded transmission: 
\begin{equation}
 X_m(i)=\gamma(i) \alpha_T U_m(i), \,\,\, M \geq m\geq 1
 \label{rand}
 \end{equation}
where $\gamma(i)$ is i.i.d. Bernoulli ($\frac{1}{2}$) over the alphabet $\{-1,1\}$
$$\gamma(i)\sim Bern (\frac{1}{2}).$$
The optimal jamming function (for adversarial sensors) is to generate i.i.d. Gaussian output
 $$ X_k(i)=\theta(i), \,\,\,  M+K \geq k\geq M+1 $$
 where 
 $$\theta(i)\sim \mathcal N(0,  P_A),$$
and is independent of  the adversarial sensor input $U_k(i)$.
The optimal receiver is the Bayesian estimator of $S$ given $Y$, i.e.,
\begin{equation}
h(Y(i))=\frac{M \alpha_T \sigma_S^2}{ { M\alpha_T^2 \sigma_{S}^2+M^2\alpha_T^2 \sigma_{W_T}^2  +K^2P_A+\sigma_Z^2}}Y(i).
\label{dec2}
 \end{equation}
Cost at this saddle-point is 
\begin{align}
J_2=\sigma_S^2\frac{M^2\alpha_T^2 \sigma_{W_T}^2+K^2P_A+\sigma_Z^2}{ M\alpha_T^2 \sigma_{S}^2+M^2\alpha_T^2 \sigma_{W_T}^2  +K^2P_A+\sigma_Z^2}
\label{cost2}
\end{align}
where $\alpha_T=\sqrt{\frac{P_T}{\sigma_S^2+\sigma_{W_T}^2}} $.
\label{th2}
\end{theorem}
\begin{proof}
We prove this result by verifying that the mappings in this theorem satisfy the saddle-point criterion given in (\ref{saddle}), following the approach in \cite{basar1986solutions}.

RHS of (\ref{saddle}): Suppose the policy of the adversarial sensors is given as in Theorem \ref{th2}. Then, the communication system at hand becomes identical to the problem considered in \cite{gastpar2008uncoded}, whose solution is uncoded communication with deterministic, linear encoders, i.e., $ X_m(i)= \alpha_T U_m(i)$. Any probabilistic encoder, given in the form of (\ref{rand}) (irrespective of the density of $\gamma$) yield the same cost  (\ref{cost2}) with deterministic encoders and hence is optimal. 

LHS of (\ref{saddle}): Note that all the adversarial sensors must use the same jamming strategy to maximize the overall cost. Let us derive the overall cost conditioned on the realization of the transmitter mappings (i.e., $\gamma=1$ and $\gamma=-1$) used in conjunction with optimal linear decoders given in (\ref{dec2}). If $\gamma=1$
\begin{equation}
D_1=J_1 +\xi \mathbb E\{SX_k\}+\psi \mathbb E\{ZX_k\}
\label{c1}
\end{equation}
  for some constants $\xi, \psi $, and similarly if  $\gamma=-1$
 \begin{equation}
D_2=J_1 -\xi \mathbb E\{SX_k\}-\psi  \mathbb E\{ZX_k\}
\label{c2}
\end{equation}
where the overall cost is
\begin{equation}
D(i)=\mathbb P(\gamma(i)=1)D_1+\mathbb P(\gamma(i)=-1)D_2 .
\end{equation} 
Clearly,  for $\gamma(i)\sim Bern (\frac{1}{2})$ the overall cost $J_1$  is only a function of the second-order statistics of the adversarial outputs, irrespective of the distribution of $\{\theta(i)\}$, and hence the solution presented here is indeed a saddle-point.
\end{proof}

\begin{cor}
The solution in Theorem 1 is (almost surely) the unique solution for the transmitters, the adversaries (jammer) and the receiver. 
\end{cor}

\begin{proof}
We start by restating the fact that the optimal solution for transmitter sensors must be in the randomized form given in (\ref{rand}). Let us prove the properties which were not covered by the proof of the saddle-point. 

Gaussianity of $\{\theta(i)\}$: The choice $\theta(i)\sim \mathcal N(0,  P_A)$ maximizes (\ref{costf}) since  it renders the simple uncoded linear mappings asymptotically optimal, i.e., the transmitters  cannot improve on the zero-delay performance by utilizing asymptotically high delays. Moreover, the optimal zero-delay performance is always lower bounded by the performance of the linear mappings, which is imposed by the adversarial choice of $\theta(i)\sim \mathcal N(0,P_A)$. 

Independence of $\{\theta(i)\}$ of $\{S(i)\}$ and $\{W(i)\}$: If the adversarial sensors introduce some correlation, i.e., if $\mathbb E\{SX_k\}\neq0$ or $\mathbb E\{WX_k\}\neq0$, the transmitter can adjust its Bernoulli parameter to decrease the distortion. Hence, the optimal adversarial strategy is setting $\mathbb E\{SX_k\}=\mathbb E\{WX_k\}=0$ which implies independence since all variables are jointly Gaussian. 

Choice of Bernoulli parameter: Note that the optimal choice of the Bernoulli parameter for the transmitters is $\frac{1}{2}$ since other choices will not cancel the cross terms  in (\ref{c1}) and (\ref{c2}), i.e., $\mathbb E\{SX_k\}$ and $\mathbb E\{WX_k\}$. These cross terms can be exploited by the adversary to increase the cost, hence optimal strategy for transmitter is to set $\gamma=Bern(1/2)$.

\end{proof}

\begin{cor}
Coordination is essential for adversarial sensors in the case of coordinating transmitters and receiver, in the sense that  lack of adversarial coordination strictly decreases the overall cost. 
\end{cor}

\begin{proof}
Note that coordination enables adversarial sensors to create a Gaussian noise with variance $K^2P_A$  yielding the cost in (\ref{cost2}). However, without coordination, each sensor can only generate independent Gaussian random variables, yielding an overall Gaussian noise with variance $KP_A$ and the total cost 
\begin{equation}J_{2}=\sigma_S^2\frac{M^2\alpha_T^2 \sigma_{W_T}^2+KP_A+\sigma_Z^2}{ M\alpha_T^2 \sigma_{S}^2+M^2\alpha_T^2 \sigma_{W_T}^2  +KP_A+\sigma_Z^2} < J_1
\end{equation} 
Hence, coordination of adversarial sensors strictly increases the overall cost, i.e., coordination is essential for adversarial sensors in this setting.
\end{proof}

\begin{remark}
Note that the optimality of this jamming strategy does not depend on the ``symmetry" assumption for the adversaries. Hence, it is straightforward to show that in the more general setting of $\sigma^2_{W_{k_1} }\neq \sigma^2_{W_{k_2} }$ and $ P_{k_1} \neq P_{k_2}$, for $ (k_1, k_2)\in [M+1\!:\!M\!+\!K]\ $ our results hold. We do not include such generalizations for brevity.
\end{remark}

\begin{remark}
We note that the optimal strategies do not depend on the sensor index $m$, hence the implementation of the optimal strategy, for both transmitter and adversarial sensors, requires ``coordination" among the sensors. This highlights the need for coordination in game theoretic settings in sensor networks. Note that this coordination requirement arises purely from the game theoretic considerations, i.e., the presence of adversarial sensors. In the case where no adversarial node exists, transmitters do not need to ``coordinate". Moreover, as we will show in Theorem 2 if the transmitters cannot coordinate, then adversarial sensors do not need to coordinate. 
\end{remark}

%

\section{Problem Setting II}
In this section, we focus on the problem, where the transmitters do not have the ability to  secretly agree on a ``coordination" random variable to generate their transmission function $X_k$. In this case, it turns out  that the optimal transmitter strategy, which is almost surely unique, is uncoded transmission with linear mappings, while  the adversarial  optimal strategy for the (jamming) sensors is uncoded, linear mappings with the opposite sign of the transmitter functions. 


The following theorem captures this result.

\begin{theorem}
Under scenario 2, the optimal encoding function for the transmitters is uncoded transmission, i.e., 
$$ X_m(i)=\alpha_T U_m(i), \,\,\, M \geq m\geq 1 $$ 

The optimal jamming function (for adversarial sensors) is uncoded transmission with the opposite sign of the transmitters, i.e.,
 $$ X_k(i)= \alpha_A U_k(i), \,\,\, M+K \geq k\geq M+1 $$

The optimal decoding function is the Bayesian estimator of $S$ given $Y$, i.e.,
$$h(Y(i))\!=\!\frac{\left [(M\alpha_T\!+\!K\alpha_A)\sigma_S^2 \right]Y(i)}{\!(M\alpha_T\!+\!K\alpha_A)\sigma_S^2\!+\!M^2\alpha_T^2\sigma_{W_T}^2\!+\!K^2\alpha_A^2\sigma_{W_A}^2\!+\!\sigma_Z^2}.$$ 
Cost at this saddle-point is 
\begin{align}
J_3=\sigma_S^2\frac{M^2\alpha_T^2\sigma_{W_T}^2+K^2\alpha_A^2\sigma_{W_A}^2+\sigma_Z^2}{(M\alpha_T+K\alpha_A)\sigma_S^2+M^2\alpha_T^2\sigma_{W_T}^2+K^2\alpha_A^2\sigma_{W_A}^2+\sigma_Z^2}
\end{align}
where $\alpha_T=\sqrt{\frac{P_T}{\sigma_S^2+\sigma_{W_T}^2}} $ and $\alpha_A=-\sqrt{\frac{P_A}{\sigma_S^2+\sigma_{W_A}^2}}$.
\end{theorem}
\begin{proof}
First, we note that adversarial sensors have the knowledge of the transmitter encoding functions, and hence the adversarial encoding functions will be in the same form as the transmitters functions but with a negative sign i.e., $\boldsymbol g_A(\cdot)=- \sqrt \frac{P_A}{P_T} \boldsymbol g_T(\cdot)$ since outputs are sent over an additive channel (see e.g., \cite{basar1986solutions,bansal1989communication} for a proof of this result). We next proceed to find the optimal encoding functions for the transmitters, given $\boldsymbol g_A(\cdot)=- \sqrt \frac{P_A}{P_T} \boldsymbol g_T(\cdot)$. From the data processing theorem, we must have 
\begin{equation}
\label{eq}
I(\boldsymbol U_1, \boldsymbol U_2, \ldots, \boldsymbol U_{M+K}; \boldsymbol {\hat S}) \leq  I(\boldsymbol X_1, \boldsymbol X_2, \ldots, \boldsymbol X_{M+K}; \boldsymbol Y) 
\end{equation}
 The left hand side can be lower bounded as:
\begin{align}
I(\boldsymbol U_1, \boldsymbol U_2, \ldots,\boldsymbol U_{M+K}; \boldsymbol {\hat S})\geq  R(D)
\end{align}
where $R(D)$ is derived in Appendix A. 
The right hand side can be upper bounded by 
\begin{IEEEeqnarray}{rCl}                                                                 
 I(\boldsymbol X_1,& \boldsymbol X_2& ,..., \boldsymbol X_{M+K}; \boldsymbol Y)   \\
& \stackrel{\text{ (a)}}{\leq} & \sum \limits_{i=1}^{N}I(X_1(i), \ldots , X_{M+K}(i);Y(i)) \\
 &\leq  & \max \sum \limits_{i=1}^{N}I(X_1(i), \ldots , X_{M+K}(i);Y(i)) \label{aa}\\
 &=  & \frac{1}{2} \sum \limits_{i=1}^{N} \log ( 1+\frac{1}{\sigma_Z^2} { \boldsymbol 1}^T R_X(i) { \boldsymbol 1} )\label{bb}
 \end{IEEEeqnarray}
  where  $R_X(i)$ is  defined as 
\begin{equation}
\{R_X(i)\}_{p,r} \triangleq\mathbb E\{X_p(i)X_r(i)\} \quad \forall p,r \in [1:M\!+\!K].
\end{equation}
Note that (a) follows from the memoryless property of the channel and the maximum in (\ref{aa}) is over the joint density over $X_{1}(i), \ldots,X_{M+K}(i)$ given the structural constraints on $R_X(i)$ due to the power constraints and also the fact that $\boldsymbol g_A(\cdot)=- \sqrt \frac{P_A}{P_T} \boldsymbol g_T(\cdot)$. It is well known that the maximum is achieved by the jointly Gaussian density for a given fixed covariance structure \cite{diggavi2001worst}, yielding (\ref{bb}).
Since logarithm is a monotonically increasing function, hence the optimal encoding functions $g_{m}^N(\cdot), m\in [1\!:\!M]$  equivalently maximize $\sum \limits_{p,r} \mathbb E\{X_p(i)X_r(i)\}$. Note that 
\begin{equation}
X_m(i)=g_{m}^N(\boldsymbol U_m)
\end{equation}
and hence  $g_{m}^N(\cdot), m\in [1\!:\!M]$ that maximize 
\begin{equation}
\sum \limits_{p=1}^{p=M+K}\sum \limits_{r=1}^{r=M+K}  \mathbb E\{g_{p}^N(\boldsymbol U_p)g_{r}^N(\boldsymbol U_r)\}
\end{equation}
 can be found by invoking  Witsenhausen's lemma (given in Appendix B) as  $g_{m}^N(\boldsymbol U_m)=\boldsymbol \alpha_N \boldsymbol U_m$ where $\boldsymbol \alpha_N=[\alpha_T,\ldots,\alpha_T]$. Finally, we obtain $J_3$ as an outer bound by equating the left and right hand sides of (\ref{eq}). The linear mappings in Theorem 2 achieve this outer bound and hence are optimal. 


 \end{proof}

\begin{cor}
Source-channel separation, based on digital compression and communications is strictly suboptimal for this setting. 
\end{cor}
\begin{proof}
We first note that the optimal adversarial encoding functions must be the negative of that of the transmitters to achieve the saddle-point solution derived in Theorem 2. But then, the problem at hand becomes equivalent to a problem with no adversary which was studied in \cite{gastpar2005power}, where source-channel separation was shown to be strictly suboptimal. Hence, separate source-channel coding has to be suboptimal for our problem.  A more direct proof follows from the calculation of the separate source-channel coding performance.
\end{proof}

\begin{cor}
Coordination is essential for transmitter sensors, in the sense that  lack of  coordination strictly increases the overall cost. 
\end{cor}
\begin{proof}
Proof follows from the fact that $J_1<J_3$.
\end{proof}

\section{Discussion and Conclusion}
This paper is an initial attempt to analyze game theoretic source-channel coding within sensor networks. We have conducted a game-theoretical analysis of a specific Gaussian sensor network, specialized to symmetric transmitter and adversarial sensors.  Depending on the {\it coordination} capabilities of the sensors, we have analyzed two problem settings. The first setting allows coordination among the transmitter sensors. Coordination capability enables the transmitters to use randomized encoders. The saddle-point solution to this problem is randomized uncoded transmission for  the transmitters and the coordinated generation of i.i.d. Gaussian noise for the adversarial sensors. In the second setting, transmitter sensors cannot coordinate, and hence they use fixed, deterministic mappings. The solution to this problem is shown to be uncoded communication with linear mappings for both the transmitter and the adversarial sensors, but with opposite signs. We note that coordination aspect of the problem is entirely due to game-theoretic considerations, i.e., if no adversarial sensors exist, optimal transmitter encoding functions do not need coordination.

Our analysis has uncovered an interesting result regarding coordination among transmitter nodes, and among adversarial nodes. If transmitter nodes can coordinate, then so must the  adversaries, i.e., all must generate the identical realization of  an i.i.d. Gaussian noise sequence. If transmitters cannot coordinate, adversarial sensors do not need to coordinate, and this saddle-point is at strictly higher cost than the one when transmitters can coordinate. 

Several questions still remain open and are currently under investigation, including  extensions of the analysis to vector sources and channels, which require a vector form of Witsenhausen's Lemma, an important research question in its own right; the  asymptotic (in the number of sensors $M$ and $K$) analysis of the results presented here; and extension of our analysis to asymmetric sensor networks and non-Gaussian settings. 
 

\begin{appendices}
\section{The Gaussian CEO Problem}
In the Gaussian CEO problem, an underlying Gaussian source $S \sim \mathcal N(0,\sigma_S^2) $ is observed under additive noise $\boldsymbol W \sim \mathcal N(\boldsymbol 0, R_W)$ as $\boldsymbol U=S+\boldsymbol W$. These noisy observations, i.e., $\boldsymbol U$, must be encoded in such a way that the decoder produces a good approximation to the original underlying source. This problem was proposed in \cite{viswanathan1997quadratic} and solved in \cite{oohama2005rate} (see also \cite{oohama1998rate,prabhakaran2004rate}). A lower bound for this function for the non-Gaussian sources within the ``symmetric" setting where all $U$'s have identical statistics was presented in \cite{gastpar2005lower}. Here, we simply extend the results in \cite{oohama1998rate} to asymmetric settings, following the approach in \cite{gastpar2005lower}, focusing on MSE distortion measure
\begin{equation}
D=\mathbb E \{(S-\hat S)^2\}
\label{eq4}
\end{equation}
 and  rate
\begin{align}
R=\min I(\boldsymbol U, \hat S)
\label{cc}
\end{align}
where $\boldsymbol U=S+\boldsymbol W$,   $\boldsymbol W \sim \mathcal N(\boldsymbol 0, R_W)$, and $R_W$ is an  $(M\!+\!K) \times (M\!+\!K)$ diagonal matrix where the first $M$ diagonal entries  are $\sigma^2_{W_T}$ and the remaining $K$  are $  \sigma^2_{W_A}$.
The minimization in (\ref{cc}) is over all conditional densities $p(\hat s|\boldsymbol u)$ that satisfy (\ref{eq4}). Distortion can be written as sum of two terms 
\begin{align}
D=&\mathbb E\{(S-T+ T-\hat S)^2\}, \\
    = &\mathbb E\{(S-T)^2\} +\mathbb E\{(T-\hat S)^2\},
    \label{eq5}
\end{align}
where $T\triangleq \mathbb E \{S| \boldsymbol U\}$. Note that  (\ref{eq5}) holds since 
\begin{equation}
\mathbb E\{(S-T)(\hat S-T)\}=0,
\end{equation} 
as the MMSE error is orthogonal to any function\footnote{Note that $\hat S$ is also a deterministic function of $\boldsymbol U$, since optimal reconstruction can always be achieved by deterministic codes.} of the observation, $\boldsymbol U$. The first term in (\ref{eq5}), i.e., the estimation error $D_{est} \triangleq \mathbb E\{(S-T)^2\} $ is constant with respect to $p(\hat s|\boldsymbol u)$, i.e., a fixed function of $\boldsymbol U$ and $S$. Hence, the minimization is over the densities that satisfy a distortion constraint of the form $\mathbb E\{(T-\hat S)^2\} \leq D_{rd}$ and  $R=\min I(\boldsymbol U; \hat S)$. Hence, we write (\ref{eq5}) as
\begin{equation}
D=D_{rd}+D_{est}.
\end{equation}
Note that due to their Gaussianity, $T$ is sufficient statistics of $\boldsymbol U$ for $S$, i.e., $S-T-\boldsymbol U$ forms a Markov chain in that order and  $T\sim \mathcal N(0,\sigma_T^2)$. Hence, $R=\min I(\boldsymbol U; \hat S)=\min I(T; \hat S)$ where minimization is over $p({\hat s}|t)$ that satisfy $\mathbb E\{(T-\hat S)^2\} \leq D_{rd}$, where all variables are Gaussian. This is the classical Gaussian rate distortion problem, and hence:
\begin{equation}
R=\frac{1}{2} \log (\sigma_T^2/D_{rd}).
\end{equation}
Noting that
\begin{equation}
T=\mathbb E \{S \boldsymbol U^*\}  (\mathbb E \{\boldsymbol U \boldsymbol U^*\} )^{-1} \boldsymbol U,
\end{equation}
and using standard linear estimation principles, we obtain 
\begin{equation}
\sigma_T^2=\sigma_S^2 \frac{ 1 }  {1+ \frac{\sigma_{W_A}^2 \sigma_{W_T}^2}{\sigma_S^2 \left(K  \sigma_{W_T}^2+M  \sigma_{W_A}^2 \right)}  },
\end{equation}
and 
\begin{equation}
D_{est}=\sigma_S^2 \frac{ \sigma_{W_T}^2 \sigma_{W_A}^2}{K  \sigma_{W_T}^2+M  \sigma_{W_A}^2+\sigma_{W_T}^2 \sigma_{W_A}^2 }.
\end{equation}


\section{Witsenhausen's Lemma}

In this section, we recall Witsenhausen's lemma \cite{witsenhausen1975sequences}, which is used in the proof of Theorem 2.

\begin{lemma}
Consider two sequences of i.i.d.  random variables $X(i)$ and $Y(i)$, generated from a joint density $P_{X,Y}$, and two (Borel measurable) arbitrary functions $f,g:\mathbb R\rightarrow \mathbb R$ satisfying 
\begin{IEEEeqnarray}{rrcll}
\mathbb E\{f(X)\}&=&\mathbb E\{g(Y)\}&=&0, \\
\mathbb E\{f^2(X)\}&=&\mathbb E\{g^2(Y)\}&=&1.
\end{IEEEeqnarray}  
Define 
\begin{equation}
\rho^*\triangleq\sup_{f,g} \mathbb E\{ f(X) g(Y)\}
\end{equation} 
Then for any (Borel measurable) functions $f_N, g_N:\mathbb R^N\rightarrow \mathbb R$ satisfying  
\begin{IEEEeqnarray}{rcl}
\mathbb E\{f_N(\boldsymbol X)\}&=&\mathbb E\{g_N(\boldsymbol Y)\}=0 ,\\
\mathbb E\{f_N^2(\boldsymbol X)\}&=&\mathbb E\{g_N^2(\boldsymbol Y)\}=1,
\end{IEEEeqnarray}
for  length $N$ vectors $\boldsymbol X$ and $\boldsymbol Y$, we have 
\begin{IEEEeqnarray}{rcl}
\sup_{f_N,g_N} \mathbb E\{f_N(\boldsymbol X )g_N(\boldsymbol Y)\} &\,\,\,\leq\,\,\, & \rho^*.
\end{IEEEeqnarray}
Moreover, the supremum above is attained by linear mappings, if  $P_{X, Y}$ is a bivariate normal density. 
\end{lemma}


\end{appendices}

\bibliographystyle{IEEEbib}

\bibliography{ref}

\end{document}